\documentclass{article}
\usepackage[T1]{fontenc}
\usepackage[utf8]{inputenc}
\usepackage{amsmath,amsfonts,amssymb,amsthm,amscd}
\usepackage{fullpage}
\usepackage{xfrac}
\usepackage{aliascnt}
\usepackage{cleveref}
\usepackage{authblk}

\usepackage{todonotes}
\usepackage{csquotes}

\usepackage{enumitem}
\setlist[description]{leftmargin=\parindent,labelindent=\parindent}

\usepackage[utf8]{inputenc}

\newtheorem{theorem}{Theorem}[section]
\newtheorem*{theorem*}{Theorem}
\newtheorem*{lemma*}{Lemma}

\newaliascnt{lemma}{theorem}
\newtheorem{lemma}[lemma]{Lemma}
\aliascntresetthe{lemma}
\crefname{lemma}{Lemma}{Lemmas}

\newaliascnt{claim}{theorem}
\newtheorem{claim}[claim]{Claim}
\aliascntresetthe{claim}
\crefname{claim}{Claim}{Claims}

\newaliascnt{corollary}{theorem}
\newtheorem{corollary}[corollary]{Corollary}
\aliascntresetthe{corollary}
\crefname{corollary}{Corollary}{Corollaries}

\newaliascnt{construction}{theorem}

\aliascntresetthe{construction}
\crefname{construction}{Construction}{Constructions}

\newaliascnt{fact}{theorem}

\aliascntresetthe{fact}
\crefname{fact}{Fact}{Facts}

\newaliascnt{proposition}{theorem}
\newtheorem{proposition}[proposition]{Proposition}
\aliascntresetthe{proposition}
\crefname{proposition}{Proposition}{Propositions}

\newaliascnt{conjecture}{theorem}

\aliascntresetthe{conjecture}
\crefname{conjecture}{Conjecture}{Conjectures}

\newaliascnt{definition}{theorem}
\newtheorem{definition}[definition]{Definition}
\aliascntresetthe{definition}
\crefname{definition}{Definition}{Definitions}

\newaliascnt{notation}{theorem}

\aliascntresetthe{notation}
\crefname{notation}{Notation}{Notation}

\newaliascnt{remark}{theorem}
\newtheorem{remark}[remark]{Remark}
\aliascntresetthe{remark}
\crefname{remark}{Remark}{Remarks}

\newaliascnt{observation}{theorem}

\aliascntresetthe{observation}
\crefname{observation}{Observation}{Observations}

\crefname{algorithm}{Algorithm}{Algorithms}

\newcommand{\N}{\ensuremath{\mathbb{N}}}
\newcommand{\negl}{\operatorname{\mathsf{negl}}}
\def\M{\mathcal{M}}
\newcommand{\SD}[2]{\operatorname{SD}\left(#1,#2\right)}

\begin{document}

\title{Stronger Lower Bounds for Online ORAM\thanks{This research was supported in part by the Grant Agency of the Czech Republic under the grant agreement no. 19-27871X, by the Charles University projects PRIMUS/17/SCI/9 and UNCE/SCI/004, Charles University grant SVV-2017-260452, and by the Neuron Fund for the support of science.}}

\author{Pavel Hubáček}
\author{Michal Koucký}
\author{Karel Král}
\author{Veronika Slívová}
\affil{Computer Science Institute of Charles University, Prague, Czech Republic \\
	\texttt{\{hubacek, koucky, kralka, slivova\}@iuuk.mff.cuni.cz}}
\date{}


\maketitle

\begin{abstract}
Oblivious RAM (ORAM), introduced in the context of software protection by Goldreich and Ostrovsky [JACM'96], aims at obfuscating the memory access pattern induced by a RAM computation.
Ideally, the memory access pattern of an ORAM should be independent of the data being processed.
Since the work of Goldreich and Ostrovsky, it was believed that there is an inherent $ \Omega(\log n) $ bandwidth overhead in any ORAM working with memory of size $ n $.
Larsen and Nielsen [CRYPTO'18] were the first to give a general $ \Omega(\log n) $ lower bound for any \emph{online} ORAM, i.e., an ORAM that must process its inputs in an online manner. 

In this work, we revisit the lower bound of Larsen and Nielsen, which was proved under the assumption that the adversarial server knows exactly which server accesses correspond to which input operation.
We give an $\Omega(\log n) $  lower bound for the bandwidth overhead of any online ORAM even when the adversary has no access to this information. For many known constructions of ORAM this information
is provided implicitly as each input operation induces an access sequence of roughly the same length.
Thus, they are subject to the lower bound of Larsen and Nielsen.
Our results rule out a broader class of constructions and specifically, they imply that obfuscating the boundaries between the input operations does not help in building a more efficient ORAM.

As our main technical contribution and to handle the lack of structure, we study the properties of \emph{access graphs} induced naturally by the memory access pattern of an ORAM computation.
We identify a particular graph property that can be efficiently tested and that all access graphs of ORAM  computation must satisfy with high probability.
This property is reminiscent of the Larsen-Nielsen property but it is substantially less structured; that is, it is more generic.

\end{abstract}

\newpage

\section{Introduction}\label{sec:Introduction}
Oblivious simulation of RAM machines, initially studied in the context of software protection by Goldreich and Ostrovsky~\cite{GoldreichO96}, aims at protecting the memory access pattern induced by computation of a RAM from an eavesdropper.
In the present day, such oblivious simulation might be needed when performing a computation in the memory of an untrusted server.\footnote{Protecting the memory access of a computation is particularly relevant in the light of the recent Spectre~\cite{KocherGGHHLMPSY18} and Meltdown~\cite{Lipp0G0HFHMKGYH18} attacks.}
Despite using encryption for protecting the content of each memory cell, the memory access pattern might still leak sensitive information.
Thus, the memory access pattern should be \emph{oblivious} of the data being processed and, optimally, depend only on the size of the input.

\paragraph{Constructions.}
The strong guarantee of obliviousness of the memory access pattern comes at the cost of additional overhead.
A trivial solution which scans the whole memory for each memory access induces linear \emph{bandwidth overhead}, i.e., the multiplicative factor by which the length of a memory access pattern increases in the oblivious simulation of a RAM with $ n $ memory cells.
Given its many practical applications, an important research direction is to construct an ORAM with as low overhead as possible.
The foundational work of Goldreich and Ostrovsky~\cite{GoldreichO96} already gave a construction with bandwidth overhead $ O(\log^3(n)) $.
Subsequent results introduced various improved approaches for building ORAMs (see~\cite{Ajtai10,ChungLP14,ChungP13a,DamgardMN11,GentryGHJ0W13,GoldreichO96,GoodrichM11,GoodrichMOT11,KushilevitzLO12,PatelP0Y18,RenFKSSDD14,StefanovDSCFRYD18,WangCS15,WangHCSS14} and the references therein) leading to the recent construction of Asharov~et~al.~\cite{AsharovKLNS18} with bandwidth overhead $ O(\log n)$ for the most natural setting of parameters.

\paragraph{Lower-bounds.} 
It was a folklore belief that an $ \Omega(\log n )$ bandwidth overhead is inherent based on a lower bound presented already in the initial work of Goldreich and Ostrovsky~\cite{GoldreichO96}.
However, the Goldreich-Ostrovsky result was recently revisited in the work of Boyle and Naor~\cite{BoyleN16}, who pointed out that the lower bound actually holds only in a rather restricted ``balls and bins'' model where the ORAM is not allowed to read the content of the data cells it processes.
In fact, Boyle and Naor showed that any general lower bound for \emph{offline} ORAM (i.e., where each memory access of the ORAM can depend on the whole sequence of operations it needs to obliviously simulate) implies non-trivial lower bounds on sizes of sorting circuits which seem to be out of reach of the known techniques in computational complexity.
The connection between offline ORAM lower bounds and circuit lower bounds was extended to \emph{read-only online} ORAMs (i.e., where only the read operations are processed in online manner) by Weiss and Wichs~\cite{WeissW18} who showed that lower bounds on bandwidth overhead for read-only online ORAMs would imply non-trivial lower bounds for sorting circuits or locally decodable codes.

The first general $ \Omega(\log n) $ lower bound for bandwidth overhead in \emph{online} ORAM (i.e., where the ORAM must process sequentially the operations it has to obliviously simulate) was given by Larsen and Nielsen~\cite{LarsenN18}.
The core of their lower bound comprised of adapting the \emph{information transfer} technique of Patrascu and Demaine~\cite{PatrascuD06}, originally used for proving lower bounds for data structures in the cell probe model, to the ORAM setting.
In fact, the lower bound of Larsen and Nielsen~\cite{LarsenN18} for ORAM can be cast as a lower bound for the oblivious Array Maintenance problem and it was recently extended to other oblivious data structures by Jacob et al.~\cite{JacobLN19}.

\subsection{Our Results}\label{sec:OurResults}
In this work, we further develop the information transfer technique of \cite{PatrascuD06} when applied in the context of online ORAMs. 
We revisit the lower bound of Larsen and Nielsen which was proved under the assumption that the adversarial server knows exactly which server accesses correspond to each input operation.
Specifically, we prove a stronger matching lower bound in a relaxed model without any restriction on the format of the access sequence to server memory.

Note that the \cite{LarsenN18} lower bound does apply to the known constructions of ORAMs where it is possible to implicitly separate the accesses corresponding to individual input operations -- since each input operation generates an access sequence of roughly the same length.
However, the \cite{LarsenN18} result does not rule out the possibility of achieving sub-logarithmic overhead in an ORAM which obfuscates the boundaries in the access pattern (e.g. by translating input operations into variable-length memory accesses).
We show that obfuscating the boundaries between the input operations does not help in building a more efficient ORAM.
In other words, our lower bound justifies the design choice of constructing ORAMs where each input operation is translated to roughly the same number of probes to server memory (common to the known constructions of ORAMs).

Besides online ORAM (i.e., the oblivious Array Maintenance problem), our techniques naturally extend to other oblivious data structures and allow to generalize also the recent lower bounds of Jacob~et~al.~\cite{JacobLN19} for oblivious stacks, queues, deques, priority queues and search trees.

For online ORAMs with statistical security, our results are stated in the following informal theorem.

\begin{theorem}[Informal]\label{thm:main_informal}
	Any statistically secure online ORAM with internal memory of size~$m$ has expected bandwidth overhead $\Omega(\log n)$, where $n \geq m^2$ is the length of the sequence of input operations.
	This result holds even when the adversarial server has no information about boundaries between probes corresponding to different input operations.
\end{theorem}

In the computational setting, we consider two definitions of computational security.
Our notion of \emph{weak computational security} requires that no polynomial time algorithm can distinguish access sequences corresponding to any two input sequences of the same length -- this is closer in spirit to computational security for ORAMs previously considered in the literature.
The notion of \emph{strong computational security} requires computational indistinguishability even when the distinguisher is given the two input sequences together with an access sequence corresponding to one of them.
The distinguisher should not be able to tell which one of the two input sequences produced the access sequence.
Interestingly, our technique (as well as the proof technique of \cite{LarsenN18} in the model with structured access pattern) yields different lower bounds with respect to the two definitions stated in the following informal theorem.

\begin{theorem}[Informal]\label{thm:main_comutational_informal}
	Any weakly computationally secure online ORAM with internal memory of size~$m$ must have expected bandwidth overhead $\omega(1)$.
	Any strongly computationally secure online ORAM with internal memory of size~$m$ must have expected bandwidth overhead $\Omega(\log n)$, where $n \geq m^2$ is the length of the sequence of input operations.
	This result holds even when the adversarial server has no information about boundaries between probes corresponding to different input operations.
\end{theorem}

Note that even the $\omega(1)$ lower bound for online ORAMs satisfying weak computational security is an interesting result in the light of the work of Boyle and Naor~\cite{BoyleN16}.
It follows from \cite{BoyleN16} that any super-constant lower bound for \emph{offline} ORAM would imply super-linear lower bounds on size of sorting circuits -- which would constitute a major breakthrough in computational complexity (for additional discussion, see Section~\ref{sec:AlternativeDefinitions}).
Our techniques clearly do not provide lower bounds for offline ORAMs.
On the other hand, we believe that proving the $\omega(1)$ lower bound in any meaningful weaker model would amount to proving lower bounds for offline ORAM or read-only online ORAM which would have important implications in computational complexity.

\paragraph{Alternative Definitions of ORAM.}
Previous works considered various alternative definitions of ORAM.
We clarify the ORAM model in which our techniques yield a lower bound in Section~\ref{sec:PrelimOram} and discuss its relation to other models in  Section~\ref{sec:AlternativeDefinitions}.
As an additional contribution, we demonstrate an issue with the definition of ORAM appearing in Goldreich and Ostrovsky~\cite{GoldreichO96}.
Specifically, we show that the definition can be satisfied by a RAM with constant overhead and no meaningful security.
The definition of ORAM in Goldreich and Ostrovsky~\cite{GoldreichO96} differs from the original definition in Goldreich~\cite{Goldreich87} and Ostrovsky~\cite{Ostrovsky90}, which do not share the issue we observed in the definition from Goldreich and Ostrovsky~\cite{GoldreichO96}.
Given that the work of Goldreich and Ostrovsky~\cite{GoldreichO96} might serve as a primary reference for our community, we explain the issue in~Section~\ref{sec:AlternativeDefinitions} to help preventing the use of the problematic definition in future works.

Persiano and Yeo~\cite{PersianoY19} recently adapted the chronogram technique~\cite{FredmanS89} from the literature on data structure lower bounds to prove a lower bound for \emph{differentially private RAMs} (a relaxation of ORAMs in the spirit of differential privacy~\cite{DworkMNS06} which ensures indistinguishability only for input sequences that differ in a single operation).
Similarly to the work of Larsen and Nielsen~\cite{LarsenN18}, the proof in~\cite{PersianoY19} exploits the fact that the distinguisher knows exactly which server accesses correspond to each input operation.
However, as the chronogram technique significantly differs from the information transfer approach, we do not think that our techniques would directly allow to strengthen the \cite{PersianoY19} lower bound for differentially private RAMs and prove it in the model with an unstructured access pattern.

\subsection{Our Techniques}\label{sec:OurTechniques}
The structure of our proof follows a similar blueprint as the work of Larsen and Nielsen~\cite{LarsenN18}.
However, we must handle new issues introduced by the more general adversarial model.
Most significantly, our proof cannot rely on any formatting of the access pattern, whereas Larsen and Nielsen leveraged the fact that the access pattern is split into blocks corresponding to each read/write operation.
To handle the lack of structure in the access pattern, we study the properties of the \emph{access graph} induced naturally by the access pattern of an ORAM computation.
We identify a particular graph property that can be efficiently tested and that all access graphs of ORAM  computation must satisfy with high probability.
This property is reminiscent of the Larsen-Nielsen property but it is substantially less structured; that is, it is more generic.

The access graph is defined as follows: the vertices are timestamps of server probes and there is an edge connecting two vertices if and only if they correspond to two subsequent accesses to the same memory cell. 
We define a graph property called \emph{$\ell$-dense $k$-partition}.
Roughly speaking, graphs with $\ell$-dense $k$-partitions are graphs which may be partitioned into $k$~disjoint subgraphs, each subgraph having at least $\ell$ edges.
We show that this property has to be satisfied (with high probability) by access graphs induced by an ORAM for any $ k $ and an appropriate $ \ell $.
To leverage this inherent structure of access graph towards a lower bound on bandwidth overhead, we prove that if a graph has $\frac{\ell}{k}$-dense $k$-partition for some $\ell$ and $K$ different values of $k$ then the graph must have at least $\Omega(\ell \log K)$ edges.
In Section~\ref{sec:graph_part}, we provide the formal definition of access graph and $\ell$-dense $k$-partitions and prove a lower bound on the expected number of edges for a graph that has many $\ell$-dense $k$-partitions.

In Section~\ref{sec:oram_lb}, we prove that access graphs of ORAMs have many dense partitions.
Specifically, using a communication-type argument we show that for $\Omega(n)$ values of~$k$, there exist input sequences for which the corresponding graph has $\Omega(\frac{n}{k})$-dense $k$-partition with high probability.
Applying the indistinguishability of sequences of probes made by ORAM, we get one sequence for which its access graph satisfies $\frac{n}{k}$-dense $k$-partition for $\Omega(n)$ values of $k$ with high probability.
Combining the above results from Section~\ref{sec:oram_lb} with the results from Section~\ref{sec:graph_part}, we get that the graph of such a sequence has $\Omega(n \log n)$ edges, and thus by definition, $\Omega(n \log n)$ vertices in expectation.  
This implies that the expected number of probes made by the ORAM on any input sequence of length $n$ is $\Omega(n \log n)$.

\section{Preliminaries}\label{sec:Preliminaries}

In this section, we introduce some basic notation and recall some standard definitions and results.
Throughout the rest of the paper, we let $[n]$ for $n\in\N$ to denote the set $\left\{ 1,2, \dots , n \right\}$.
A function $\negl(n)\colon \mathbb{N} \rightarrow \mathbb{R}$ is \emph{negligible} if it approaches zero faster than any inverse polynomial.

\begin{definition}[Statistical Distance]\label{def:statistical_distance}
	For two probability distributions $X$ and $ Y$ on a discrete universe $S$, we define \emph{statistical distance of $ X $ and $ Y $}  as
	$$\SD{X}{Y} = \frac{1}{2} \sum_{s \in S} | \Pr[X = s] - \Pr[Y = s]|\ .$$
\end{definition}

We use the following observation, which characterizes statistical distance as the difference of areas under the curve (see~Fact~3.1.9 in Vadhan~\cite{Vadhan99}).

\begin{proposition}\label{prop:StatDistanceAsArea}
Let $X$ and $Y$ be probability distributions on a discrete universe $S$, let $S_X=
\{s\in S\colon \Pr[X=s]>\Pr[Y=s]\}$, and define $S_Y$ analogously.
Then
$$
\SD{X}{Y}=\Pr[X\in S_X]-\Pr[Y\in S_X]=\Pr[Y\in S_Y]-\Pr[X\in S_Y]\ .
$$
\end{proposition}

We also use the following data-processing-type inequality.

\begin{proposition}\label{prop:stat_dif}
Let $X$ and $Y$ be probability distributions on a discrete universe $S$.
Then for any function $f\colon S \rightarrow \left\{ 0,1 \right\}$, it holds that $|\Pr[f(X) = 1] - \Pr[f(Y) = 1]| \leq \SD{X}{Y}$.
\end{proposition}

\begin{definition}[Computational indistinguishability]\label{def:ComputationalIndistinguishability}
	Two probability ensembles, $\left\{ X_n \right\}_{n\in\N}$ and  $\left\{ Y_n \right\}_{n\in\N}$,  are \emph{computationally indistinguishable} if for every polynomial-time algorithm $D$ there exists a negligible function $ \negl(\cdot) $ such that
	$$
	\left| \Pr[D(X_n,1^n) = 1] - \Pr[D(Y_n,1^n) = 1] \right| \leq \negl(n)\ .
	$$
\end{definition}

\subsection{Online ORAM}\label{sec:PrelimOram}

In this section, we present the formal definition for online oblivious RAM (ORAM) we consider in our work -- we build on the oblivious cell-probe model of Larsen and Nielsen~\cite{LarsenN18}.

\begin{definition}[Array Maintenance Problem~\cite{LarsenN18}]
	The \emph{Array Maintenance problem with parameters $(\ell, w)$} is to maintain an array $B$ of $\ell$ $w$-bit entries under the following two operations:
	\begin{itemize}
		\item  \emph{($W, a$, $d$)}: Set the content of $B[a]$ to $d$, where $a \in [\ell]$, $d \in \left\{ 0,1 \right\}^{w}$.
		    \hfill \emph{(Write operation)}
		\item  \emph{($R, a$, $d$)}: Return the content of $B[a]$, where $a \in [\ell]$ (note that $d$ is ignored).
		    \hfill \emph{(Read operation)}
	\end{itemize}
	We say that a machine $\M$ \emph{implements the Array Maintenance problem with parameters $(\ell, w)$ and probability $p$}, if for every input sequence of operations
	$$
		y = (o_1, a_1, d_1), \dots , (o_n, a_n, d_n) \text{, where each } o_i \in \left\{ R,W \right\}, a_i \in [\ell], d_i \in \left\{ 0,1 \right\}^w,
	$$
	and for every read operation in the sequence $y$, $\M$ returns the correct answer with probability at least $p$.
	\label{def:array_maintenance}
\end{definition}

\begin{definition}[Online Oblivious RAM]\label{def:oram}
	For $m, w\in\N$, let RAM*$(m, w)$ denote a probabilistic random access machine $\M$ with $m$ cells of internal memory, each of size $w$ bits, which has access to a data structure, called \emph{server}, implementing the Array Maintenance problem with parameters $(2^w, w)$ and probability 1.
	In other words, in each step of computation $\M$ may \emph{probe} the server on a triple $(o, a, d)\in \left\{ R,W \right\}\times [2^w]\times \left\{ 0,1 \right\}^w$ and on every input $(R, a, d)$ the server returns to $\M$ the data last written in $B[a]$.
	We say that $RAM^*$ \emph{probes} the server whenever it makes an Array Maintenance operation to the server.

	Let $m, M, w$ be any natural numbers such that $M \leq 2^w$.
	An \emph{online Oblivious RAM~$\M$ with address range $M$, cell size $w$ bits and $m$ cells of internal memory} is a $RAM^*(m, w)$ satisfying online access sequence, correctness, and statistical (resp. computational) security as defined below.

   \begin{description}
    \item[Online Access Sequence:]
    For any input sequence $y = y_1, \dots, y_n$ the RAM* machine $\M$ gets $y_i$ one by one, where each $y_i \in \left\{ R, W \right\} \times [M] \times \left\{ 0, 1 \right\}^w$.
		Upon the receipt of each operation $y_i$, the machine $\M$ generates a possibly empty sequence of \emph{server probes} $(o_1, a_1, d_1), \dots ,(o_{\ell_i}, a_{\ell_i}, d_{\ell_i})$, where each $(o_i, a_i, d_i) \in \left\{ R, W \right\} \times [2^w] \times \left\{ 0,1 \right\}^w$, and updates its internal memory state in order to correctly implement the request~$y_i$.
		We define the \emph{access sequence corresponding to $ y_i $} as $A(\M, y_i) = a_1, a_2, \ldots, a_{\ell_i}$.
    For the input sequence $y$, the \emph{access sequence} $A(\M, y)$ is defined as $$A(\M, y) = A(\M, y_1),A(\M, y_2),A(\M, y_3), \ldots , A(\M, y_n).$$

		Note that the definition of the machine $\M$ is \emph{online}, and thus for each input sequence $y = y_1, \dots, y_n$ and each $i \in [n-1]$, the access sequence $A(\M, y_i)$ does not depend on $y_{i+1},\dots,y_n$.

	\item[Correctness:]
	$\M$ implements the Array Maintenance problem with parameters $(M, w)$ with probability at least~$1 - p_{\mathrm{fail}}$.

	\item[Statistical Security:]
	For any two input sequences $y, y'$ of the same length,  the statistical distance of the distributions of access sequences $A(\M, y)$ and $A(\M, y')$ is at most $\frac{1}{4}$.
	
	\item[Computational Security:]
	For computational security, we consider infinite families of ORAM where we allow  $m, M, w$ to be functions of the length $n$ of the input sequence.	
    We distinguish between the following two notions:
    
    \begin{description}
    \item[Weak Computational Security:]\label{def:comp_sec_weak}
	For any infinite families of input sequences $\{y_n\}_{n\in\N}$ and $\{y'_n\}_{n\in\N}$ such that  $|y_n|= |y'_n|\ge n$ for all $ n\in\N $, the probability ensembles $\{A(\M, y_n)\}_{n\in\N}$ and $\{A(\M, y'_n)\}_{n\in\N}$ are computationally indistinguishable. 
	
	\item[Strong Computational Security:]\label{def:comp_sec_strong}
	For any infinite families of input sequences $\{y_n\}_{n\in\N}$ and $\{y'_n\}_{n\in\N}$ such that  $|y_n|= |y'_n|\ge n$ for all $ n\in\N $, the probability ensembles $\{(y_n, y'_n, A(\M, y_n))\}_{n\in\N}$ and
	$\{(y_n, y'_n, A(\M, y'_n))\}_{n\in\N}$ are computationally indistinguishable. 
    \end{description}
    
\end{description}

\end{definition}

The parameters of our ORAM model from Definition~\ref{def:oram} are depicted in Figure~\ref{fig:oram}.
We use different sizes of arrows on server and RAM side to denote the asymmetry of the communication (the RAM sends type of operation, address, and data and the server returns requested data in case of a read operation and dummy value in case of a write operation).
Note that the input sequence $y$ of ORAM consists of a sequence of all operations, whereas the access sequence $A(\M, y)$ consists of a sequence of addresses of all probes.

\begin{figure}
    \begin{center}
    \includegraphics[width=0.99\textwidth]{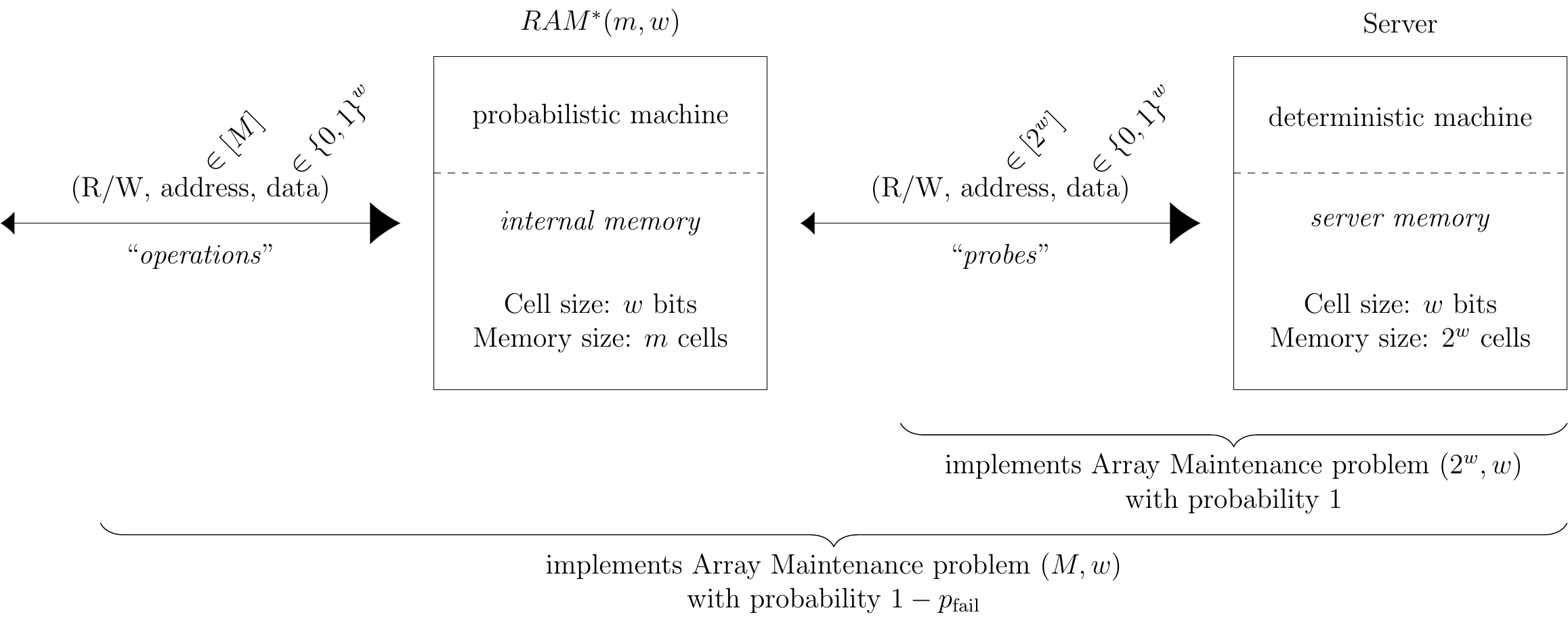}
    \caption{Schema of online ORAM from Definition~\ref{def:oram}.}
    \end{center}
    \label{fig:oram}
\end{figure}

Arguably, a user of an ORAM might want the stronger notion of computational security whereas the weaker notion is closer to the past considerations.
Note that in the case of weak computational security, the adversarial distinguisher does not have access to the input sequences.
Thus, it is restricted to contain only constant amount of information about the whole families of input sequences $\{y_n\}_n$ and $\{y'_n\}_n$.
In contrast, in the case of strong computational security, the adversarial distinguisher is given also the input sequences.
Thus, it is able to compute any polynomial time computable information about the input sequences.
This distinction is crucial for our results, as we are able to prove only an $\omega(1)$ lower bound for weak security as opposed to the $\Omega(\log n)$ lower bound for strong security (see Theorem~\ref{thm:mainomega} and Theorem~\ref{thm:main_strong_computational}).
Nevertheless, we believe that the known constructions of ORAM satisfy the notion of strong computational security.

For ease of exposition, in the rest of the paper we assume perfect correctness of the ORAM (i.e., $p_{\mathrm{fail}} = 0$).
However, our lower bounds can be extended also to ORAMs with imperfect correctness (see Remark~\ref{rem:error_correcting_codes}).
Finally, our lower bounds hold also for \emph{semi-offline} ORAMs where the ORAM machine $\M$ receives the type and address of each operation in advance and it has to process in online manner only the data to be written during each write operation (see Remark~\ref{rem:semioffline}).


\section{Dense Graphs}\label{sec:graph_part}

In this section, we define an efficiently testable property of graphs that we  show to be satisfied by graphs induced by the access pattern of any statistically secure ORAM.
This property implies that the overhead of such ORAM must be logarithmic.

We say a directed graph $G=(V,E)$ is {\em ordered} if $V$ is a subset of integers and for each edge $(u,v)\in E$, $u<v$.
For a graph $G = (V, E)$ and $S,T\subseteq V$, we let $E(S,T) \subseteq E$ be the set of edges that start in $S$ and end in $T$,
and for integers $a\le m \le b \in V$ we let $E(a,m,b)=E(\{a,a+1,\dots,m-1\},\{m,m+1,\dots,b-1\})$.

\begin{definition}
	A \emph{$k$-partition} of an ordered graph $G = (V = \{0, 1, 2, \dots, N-1\}, E)$ is a sequence $0=b_0\le m_0 \le b_1 \le m_1 \le \cdots \le b_k = N$. We say that the $k$-partition is $\ell$-dense if for each $i\in \{0,\dots,k-1\}$, $E(b_i,m_i,b_{i+1})$ is of size
	at least $\ell$.
	\label{def:dense_block_partition}
\end{definition}

There is a simple greedy algorithm running in time $\mathcal{O}(|V|^2 \cdot |E|)$ which tests for given integers $k, \ell$ whether a given ordered graph $G = (V, E)$ has an $\ell$-dense $k$-partition. (The algorithm looks for the $k$ parts one by one greedily from left to right.)

\begin{lemma}
	Let $K\subseteq \N$ be a subset of powers of 4. Let $\ell \in \N$ be given. Let $G = (\{0,\dots,N-1\}, E)$ be an ordered graph which for each $k\in K$
	has an $(\ell/k)$-dense $k$-partition. Then $G$ has at least $\frac{\ell}{2}\cdot |K|$ edges.
	\label{lemma:partitioning_implies_enough_edges}
\end{lemma}

\begin{proof}
	We use the following claim to bound the number of edges.
	\begin{claim}
		Let $k > k'>0$ be integers.
		Let $0=b_0\le m_0 \le b_1 \le m_1 \le \cdots \le b_k = N$ be a $k$-partition of $G$, and 
		$0=b'_0\le m'_0 \le b'_1 \le m'_1 \le \cdots \le b'_{k'} = N$ be a $k'$-partition of $G$. 
		Then for at least $k-k'$ distinct $i\in\{0,\dots,k-1\}$
		\begin{align}
		E(b_i,m_i,b_{i+1}) \cap \bigcup_{j\in \{0,\dots,k'-1\}} E(b'_j,m'_j,b'_{j+1})  = \emptyset.
		\label{cond:alligned}
		\end{align}
		\label{claim:middleOverlay}
	\end{claim}
	\begin{proof}
		For any $j\in \{0,\dots,k'-1\}$ and $(u,v)\in E(b'_j, m'_j,b'_{j+1})$, if $(u,v)\in E(b_i,m_i,b_{i+1})$ for some~$i$ then $b_i < m'_j  < b_{i+1}$ (as $b_i \leq u < m'_j \le v \leq b_{i+1}$.) Thus, $i$ is uniquely determined by $j$. Hence, $E(b_i,m_i,b_{i+1})$ may intersect $\bigcup_{j\in \{0,\dots,k'-1\}} E(b'_j, m'_j, b'_{j+1})$ only if $b_i \le m'_j  < b_{i+1}$, for some $j\in \{0,\dots,k'-1\}$. Thus, such an intersection occurs only for at most $k'$ different $i$. The claim follows.
	\end{proof}
	
	Now we are ready to prove Lemma~\ref{lemma:partitioning_implies_enough_edges}.
	For each $k\in K$, pick an $(\ell/k)$-dense $k$-partition $0=b_0\le m_0 \le b_1 \le m_1 \le \cdots \le b_k = N$ of $G$ and define the set of edges $E_k$:
	$$
	E_k = \bigcup_{i\in \{0,\dots,k-1\}} E(b_i, m_i,b_{i+1}).
	$$
	
	For each $k \in K$, we lower-bound $\left|E_k \setminus \bigcup_{k'\in K,k'<k} E_{k'} \right|$ by $\ell/2$. Since $K$ contains powers of~$4$, 
	$\sum_{k'\in K,k'<k} k' \le k/2$. By the above claim, for at least $k - \sum_{k'\in K,k'<k} k' \ge k/2$ different $i\in \{0,\dots,k-1\}$, 
	$E(b_i, m_i,b_{i+1}) \cap \bigcup_{k'\in K,k'<k} E_{k'} = \emptyset$. 
	By density, $|E(b_i, m_i,b_{i+1})| \ge \ell/k$, so $\left|E_k \setminus \bigcup_{k'\in K,k'<k} E_{k'}\right| \ge \frac{\ell}{k}\cdot \frac{k}{2} = \ell/2$. 
	Hence, $\left|\bigcup_{k\in K} E_k\right| = \sum_{k \in K} \left|E_k \setminus \bigcup_{k'\in K,k'<k} E_{k'}\right| \ge |K| \cdot \frac{\ell}{2}$.
\end{proof}

In the following corollary, we show that the property of having many dense partitions with some probability implies proportionally many edges. (Note that the $\lfloor\log_4 t \rfloor - \lceil\log_4 s \rceil$ term corresponds exactly to the number of powers of four between $s$ and $t$.)

\begin{corollary}
	Let $\ell, s, t$ be natural numbers, where $s \le t$. Let $p \in [0, 1]$ be a real.
	Let $G$ be an ordered graph picked at random from a distribution such that
	for each integer $k$, $s \leq k \leq t$, the randomly chosen ordered graph $G$ has $(\ell/k)$-dense $k$-partition with probability at least $p$.
	Then the expected number of edges in $G$ is at least $\frac{p \ell}{2} \cdot (\lfloor\log_4 t \rfloor - \lceil\log_4 s \rceil)$.
	\label{cor:main_graph_cor}
\end{corollary}

\begin{proof}
	Let $K$ be the set of integers such that $k \in K$ if and only if $k$ is a power of~$4$ and $G$ has an $(\ell/k)$-dense $k$-partition. $K$ is a random variable. The expected 
	size of $K$ is at least $p (\lfloor\log_4 t \rfloor - \lceil\log_4 s \rceil)$. By Lemma~\ref{lemma:partitioning_implies_enough_edges}, the expected number of edges in $G$ is at least $\frac{\ell}{2} \cdot p \cdot (\lfloor\log_4 t \rfloor - \lceil\log_4 s \rceil)$.
\end{proof}

\section{ORAM Lower Bound}\label{sec:oram_lb}

In this section, we fix integers $n,m,M,w \ge 1$ such that $m \leq \sqrt{n}$, $n \leq M \leq 2^w$,  and an ORAM $\M$ with address range $M$, cell size $w$ and $m$ cells of internal memory (see Definition~\ref{def:oram}).
We argue that any statistically secure ORAM $\M$ must make $\Omega(n \log n)$ server probes in expectation in order to implement a sequence of $n$ input operations.
We also show that any ORAM $\M$ satisfying Weak Computational Security must make $\omega(n)$ server probes in expectation on any input sequence of length $n$.

\begin{definition}
	Let $A(\M, y) = a_0, \dots, a_{N-1}$ be an access sequence of $\M$ for some input sequence $y$. 
	We define a directed graph $G(A(\M, y)) = (V, E)$ called \emph{access graph} as follows:
	$V = \{0, \dots, N-1\}$ and $(i, j) \in E$ iff $i < j$ and $a_i = a_j$ and for each $k \in \{i+1, \dots, j-1\}$, $a_k \neq a_i$.
	\label{def:probe_graph}
\end{definition}

Notice that every vertex of an access graph has outdegree as well as indegree at most one.

In the following, we consider input sequences of even length $n\in \N$.
First, we define a sequence of alternating writes and reads at address $a = 1$ with data $d = 0^w$ as $Y_{n,0} = \left[ (W, 1, 0^w), (R, 1, 0^w) \right] ^{n/2}$.
Second, for each $k \in \left\{ 1, 2, \dots, \frac{n}{2} \right\}$, let $\ell = \left\lfloor \frac{n}{2k} \right\rfloor$, we define a distribution $Y_{n,k}$ of input sequences as
\begin{align*}
Y_{n,k} =
&(W, 1, b_{1, 1}), (W, 2, b_{1, 2}), \dots, (W, \ell, b_{1, \ell}), (R, 1, 0^w), (R, 2, 0^w), \dots, (R, \ell, 0^w),\\
&(W, 1, b_{2, 1}), (W, 2, b_{2, 2}), \dots, (W, \ell, b_{2, \ell}), (R, 1, 0^w), (R, 2, 0^w), \dots, (R, \ell, 0^w),\\
&\dots,\\
&(W, 1, b_{k, 1}), (W, 2, b_{k, 2}), \dots, (W, \ell, b_{k, \ell}), (R, 1, 0^w), (R, 2, 0^w), \dots, (R, \ell, 0^w),\\
&(W, 1, 0^w), (R, 1, 0^w), (W, 1, 0^w), \dots, (R, 1, 0^w)\ ,
\end{align*}
where each $b_{i, j} \in \left\{ 0,1 \right\}^w$ is an independently uniformly chosen bit string.
We define the $i$-th block of writes $W_i = (W, 1, b_{i, 1}), (W, 2, b_{i, 2}), \dots, (W, \ell, b_{i, \ell})$
and the $i$-th block of reads $R_i$ to be the sequence of operations $(R, 1, 0^w), (R, 2, 0^w), \dots, (R, \ell, 0^w)$ following right after $W_i$.
Note that after the $k$-th block of reads the sequence is padded to length $n$ by a sequence of alternating writes and reads.
For an ORAM $\M$, we use the notation $G_{n,k} = G(A(\M, Y_{n,k}))$ and $G_{n,0} = G(A(\M, Y_{n,0}))$ when $\M$ is clear from the context.

The following lemma uses only correctness of ORAM and does not depend on its security.
The proof of the lemma uses the information transfer technique similarly to Lemma~2 in~\cite{LarsenN18}.

\begin{lemma}
	Let $n, m, M, w, \M$ be as in the beginning of this section, moreover suppose $n\ge 10$ is an even integer.
	Let $k\ge 1$ be an integer such that $k \le \frac{n}{10(m+2 \log n + 11)}$. Let $A(\M, Y_{n,k})$ be the access sequence of $\M$
	and $G_{n,k}$ be the corresponding access graph. ($G_{n,k}$ is a random variable that depends on $Y_{n,k}$ and the internal randomness of $\M$.)
	With probability at least $1-\frac{1}{n}$, $G_{n,k}$ has $(n/5k)$-dense $k$-partition.
	\label{lem:edges_inside_blocks}
\end{lemma}

\begin{proof}
    By our assumption from the beginning of this section, $n \leq M$, and thus for any $k \in \{1, 2, \ldots, \frac{n}{2}\}$ all sequences $Y_{n, k}$ have all addresses in the correct range.
	Fix any $k$ satisfying the assumptions of this lemma and set $\ell = \left\lfloor \frac{n}{2k} \right\rfloor$.
	As defined before let $W_i$ and $R_i$ be the $i$-th block of writes and reads in $Y_{n,k}$, respectively.  Let $U_i$ be the vertices of $G_{n,k}$ corresponding to $W_i$,
	and $V_i$ be the vertices corresponding to $R_i$.
	It suffices to prove that for each $i\in \{1,\dots,k\}$, the probability that there are fewer than $n/5k$ edges between $U_i$ and $V_i$ is less than $1/n^2$.
	If this holds then by the union bound the lemma follows.
	
	For contradiction, assume there exists $i\in \{1,\dots,k\}$ such that the probability that there are fewer than $n/5k$ edges between $U_i$ and $V_i$ is at least $1/n^2$.
	Here, the randomness is taken over the choice of an input sequence $y \leftarrow Y_{n,k}$ and the internal randomness of $\M$. Fix such an $i$. Fix all the randomness except for the choice of 
	$b_{i,1},\dots,b_{i,\ell}$ in $Y_{n,k}$ so that $G_{n,k}$ obtained from this restricted distribution has fewer than $n/5k$ edges between $U_i$ and $V_i$ 
	with probability $\ge 1/n^2$ over the choice of $b_{i,1},\dots,b_{i,\ell}$. (This is possible by an averaging argument.)
	Let $B \subseteq \{0,1\}^{w \times \ell}$ be the set of choices for $b_{i,1},\dots,b_{i,\ell}$ which give fewer than $n/5k$ edges between $U_i$ and $V_i$
	in $G_{n,k}$. Clearly, $|B| \ge 2^{w\ell} / n^2$.
	
	We use $\M$ to construct a deterministic protocol that transmits any string from $B$ from Alice to Bob, two communicating parties, 
	using at most $\log |B| - 10$ bits. That gives a contradiction as such an efficient transmission violates the pigeon-hole principle.
	
	On input $b\in B$ to Alice, Alice sends a single message to Bob who can determine~$b$ from the message. They proceed as follows.
	Both Alice and Bob simulate $\M$ on $Y_{n,k}$ up until reaching $W_i$.
	All the randomness used before the $i$-th block of writes $W_i$ is fixed and known both to Alice and Bob.
	Then Alice continues with the simulation of $\M$ on $W_i$ with data $b_{i, 1}, b_{i, 2}, \ldots, b_{i, \ell}$ set to $b$.
	Once she finishes it, she sends
	the content of the internal memory of $\M$ to Bob using $wm$ bits. Then Alice continues with the simulation of $\M$ on $R_i$ and whenever $\M$ 
	makes a server probe to read from a location that was written last time during the simulation of $W_i$, Alice sends over the address and the content of that cell to Bob.
	Overall, Alice sends at most $mw + 2wn/5k$ bits of communication to Bob that can be concatenated into a single message of this size.
	
	On receiving side, Bob uses the internal state of $\M$ communicated by Alice to continue with the computation on $R_i$, while he uses the state of the server
	he obtained initially before reaching~$W_i$. He simulates all server probes by himself, except for read operations that match the list sent by Alice, where
	he initially uses the content provided by Alice. Clearly, Bob can determine~$b$ from the simulation.
	
	As $k \le \frac{n}{10(m+2 \log n + 11)}$,  $mw + 2wn/5k \le \left(n/2k - 2\log n -11\right)w$, so $mw + 2wn/5k \le (\ell - 2\log n - 10)w$, hence, the number of communicated bits is $mw + 2wn/5k \le \log |B| - (2w - 2)\log n - 10w$, which is a contradiction.
\end{proof}

\begin{remark}
Using good error-correcting codes (see for instance \cite{MacWS77}), this lemma could be generalized to the case when $\M$ implements Array Maintenance problem with probability $1 - p_{\mathrm{fail}} < 1$, i.e., $\M$ is allowed to return a wrong value 
for each of its input read operations with a small constant  probability $p_{\mathrm{fail}}$.
The graph $G_{n,k}$ would still have $(\epsilon n/k)$-dense $k$-partition with $1-1/n$ probability for some $\epsilon>0$ which depends only on the allowed failure probability $p_{\mathrm{fail}}$.
\label{rem:error_correcting_codes}
\end{remark}

\begin{remark}
Note that the randomness of input sequence $Y_{n,k}$ is used only for the data to be written.
Moreover, the proof relies only on incompressibility of a random string stored during the write block and it does not rely on the addresses used to store this data.
Thus, the same proof goes through even for \emph{semi-offline ORAMs}, i.e., if we allow the ORAM to know the type and address of each input operation in $y$ in advance.
On the other hand, as our proof uses interleaved sequences of write blocks and read blocks, it is unlikely that it would be possible to extend it to the \emph{read-only online ORAM} model of Weiss and Wichs~\cite{WeissW18}.
\label{rem:semioffline}
\end{remark}

Note that using an averaging argument we can assume that the probability in Lemma~\ref{lem:edges_inside_blocks} is only over the randomness of $\M$.
Thus we get the following corollary proving for every $k$ the existence of a single input sequence whose corresponding access graph has $\frac{n}{5k}$-dense $k$-partition with high probability.

\begin{corollary}
    For any even integer $n \geq 10$ and an integer $k\ge 1$ such that $k \le \frac{n}{10(m+2 \log n + 11)}$ there is an input sequence $y_{n,k}$ of length $n$ such that $G(A(\M, y_{n,k}))$ has a $(n/5k)$-dense $k$-partition with probability at least $1-\frac{1}{n}$.
    \label{cor:hard_op_sequence}
\end{corollary}

We show that by statistical security of $\M$, this property holds for a single input sequence and many different values of~$k$.

\begin{lemma}
	Let $n, m, M, w, \M$ be as in the beginning of this section, and assume $n$ is even and $n\ge 10$. Let $y$ be an input sequence to $\M$ of length $n$.
	If $\M$ is a statistically secure online ORAM then for every $k \in \left\{ 1, 2, \dots, \left\lfloor \frac{n}{10(m+2 \log n + 11)} \right\rfloor \right\}$
	$$\Pr\left[G(A(\M,y)) \text{ has an $(n/5k)$-dense $k$-partition}\right] \geq \frac{3}{5}.$$
	\label{lemma:graph_properties_statistical_sec}
\end{lemma}

\begin{proof}
	For contradiction, suppose that for some~$k$ the probability is less than $3/5$.
	From the statistical security of $\M$ we know that the statistical distance $\SD{A(\M,y)}{A(\M, y_{n,k})} \leq \frac{1}{4}$ where $y_{n, k}$ is given by Corollary~\ref{cor:hard_op_sequence}.
	By Corollary~\ref{cor:hard_op_sequence} the sequence $y_{n, k}$ gives us a graph $G(A(\M, y_{n,k}))$ which has an $(n/5k)$-dense $k$-partition with probability at least $1-1/n \ge 9/10$.
	Define a function $f_{\ell,k}$ on ordered graphs that is an indicator of having an $\ell$-dense $k$-partition.
	Applying Proposition~\ref{prop:stat_dif} with $X \leftarrow G(A(\M,y))$, $Y \leftarrow G(A(\M,y_{n,k}))$, and $f=f_{n/5k,k}$, we can
	conclude that $G(A(\M,y))$ has an $(n/5k)$-dense $k$-partition with probability at least $3/4-1/10 \ge 3/5$.
\end{proof}

We are ready to prove our main theorem for statistically secure ORAM.

\begin{theorem}
    There are constants $c_0, c_1 > 0$ such that for any integers $m, w\ge 1$ and $M \ge n \geq c_0$  where $m \leq \sqrt{n}$ and $M \leq 2^{w}$, any statistically secure online ORAM $\M$ with     address range $M$, cell size $w$ bits and $m$ cells of internal memory must perform at least $c_1 n \log n$ server probes in expectation (the expectation is over the randomness of $\M$) on any input sequence of length $n$.
	\label{thm:main}
\end{theorem}

\begin{proof}
    Fix an ORAM machine $\M$. Consider any input sequence $y$ to $\M$ of length $n$.
    By Lemma~\ref{lemma:graph_properties_statistical_sec} for every $k$, such that
	$1 \leq k \leq \left\lfloor \frac{n}{10(m+2 \log n + 11)} \right\rfloor$,
	we get that
	$$\Pr\left[G(A(\M,y)) \text{ has an $(n/5k)$-dense $k$-partition}\right] \geq \frac{3}{5}.$$

	Applying Corollary~\ref{cor:main_graph_cor} with $s = 1$, $t = \left\lfloor \frac{n}{10(m+2 \log n + 11)} \right\rfloor$, $\ell = \left\lfloor\frac{n}{5}\right\rfloor$, and $p = 3/5$,
	we can lower bound the expected number of edges in $G(A(\M,y))$ by
	$$\frac{3n}{50}  \left\lfloor \log_4 \left\lfloor \frac{n}{10(m + 2 \log n + 11)} \right\rfloor  \right\rfloor .$$
	 For $n\ge 1000$, $\left\lfloor \frac{n}{10(m+2 \log n + 11)} \right\rfloor \ge \frac{\sqrt{n}}{40}$. 
	 Hence, the expected number of edges in $G(A(\M,y))$ is at least $\frac{3}{100} \cdot  n \log \frac{\sqrt{n}}{40} \ge \frac{1}{100} \cdot n \log n$, provided $c_0$ is large enough.
	 Since the indegree of each vertex of an access graph is at most one, the expected number of vertices in $G(A(\M,y))$, which is the same as the expected number of probes in $A(\M,y)$,  is at least $\frac{1}{100} \cdot n \log n$.
\end{proof}

Next, we prove $\Omega(\log n)$ lower bound for ORAMs satisfying strong computational security from Definition~\ref{def:comp_sec_strong}.

\begin{lemma}

Let $m, M, w \colon \mathbb{N} \rightarrow \mathbb{N}$ be non-decreasing functions such that for all $n$ large enough: $m(n) \leq \sqrt{n}$ and $n \leq M(n) \leq 2^{w(n)}$.
    Let $\{ \M_n \}_{n\in \N}$ be a sequence of online ORAMs with address range $M(n)$, cell size $w(n)$ bits and $m(n)$ cells of internal memory which satisfy strong computational security. Let $\{y_n\}_{n\in \N}$ be an infinite family of input sequences where $|y_n|=n$, for each $n\in \N$.

	Then there exists $n_0$ such that for every $n \geq n_0$ and for every $k \in \left\{ 1, 2, \dots, \left\lfloor \frac{n}{10(m(n)+2 \log n + 11)} \right\rfloor \right\}$
	$$\Pr\left[G(A(\M_{n},y_n)) \text{ has an $(n/5k)$-dense $k$-partition}\right] \geq \frac{3}{5}.$$
	\label{lemma:graph_properties_strong_comp_sec}
\end{lemma}

\begin{proof}
	For contradiction, assume there are infinitely many pairs of integers $(n,k)$, s.t. $k \leq \left\lfloor \frac{n}{10(m(n)+2 \log n + 11)} \right\rfloor$ and that the probability that $y_n$ has an $(n/5k)$-dense $k$-partition is less than $3/5$.
	
	Let $\mathcal{D}$ be an algorithm which given two input sequences $y$ and $y'$ of length $n$ and an access sequence $A(\M_{n},z)$, where $z \in \{y,y'\}$, does the following:
	\begin{enumerate}
	    \item Compute $n$.
	    \label{D:step:one}
	    \item Compute $k'$ to be the number of blocks of consecutive reads of length $\lfloor n/k' \rfloor$ in the input sequence $y'$.
	    \label{D:step:two}
	    \item If $A(\M_{n},z)$ does not have $(n/5k')$-dense $k'$-partition $\mathcal{D}$ returns~\enquote{1} (i.e. $D$ guesses that $z = y$).
	    \label{D:step:three}
	    \item Otherwise $\mathcal{D}$ returns \enquote{1} with probability $1/2$ and \enquote{2} with probability $1/2$ (i.e. $D$ guesses at random).
	    \label{D:step:four}
	\end{enumerate}
	
	There is a polynomial time greedy algorithm determining whether the graph $G(A(\M_{n}, z))$ contains an $\ell$-dense $k$-partition.
	Thus algorithm $\mathcal{D}$ runs in time polynomial in the length of the access sequence $A(\M_{n}, z)$.
	
	Let $y_{n,k}$ be a sequence from Corollary~\ref{cor:hard_op_sequence}. So, $G(A(\M_{n}, y_{n,k}))$ has an $(n/5k)$-dense $k$-partition with probability at least $1-1/n \ge 9/10$.
	Observe that if $y=y_n$ and $y' = y_{n,k}$  then:
	$$
	\left|
	\Pr[\mathcal{D}(y_n, y_{n,k}, A(\M_{n}, y_n)) = 1]
	- \Pr[\mathcal{D}(y_n, y_{n,k}, A(\M_{n}, y_{n,k})) = 1] 
	\right|
	 \geq
	\left(\frac{2}{5} + \frac{3}{5}\cdot \frac{1}{2}\right)
	- \left(\frac{1}{10} + \frac{9}{10} \cdot \frac{1}{2}\right)
	= \frac{3}{20}.
	$$
	By the assumption $\mathcal{D}$ returns \enquote{1} in step~\ref{D:step:three} on $A(\M_{n}, y_n)$ with probability at least $2/5$.
	By Corollary~\ref{cor:hard_op_sequence} $\mathcal{D}$ answers \enquote{1} on $A(\M_{n}, y_{n, k})$ with probability at most $1/10$.
	
	This contradicts the strong computational security of $\M_{n}$ as $\mathcal{D}$ should not distinguish between $y$ and $y'$ with non-negligible probability.
\end{proof}

\begin{theorem}
    
    Let $m, M, w \colon \mathbb{N} \rightarrow \mathbb{N}$ be non-decreasing functions such that for all $n$ large enough: $m(n) \leq \sqrt{n}$ and $n \leq M(n) \leq 2^{w(n)}$.
    Let $\{ \M_{n} \}_{n\in \N}$ be a sequence of online ORAMs with address range $M(n)$, cell size $w(n)$ bits and $m(n)$ cells of internal memory which satisfy strong computational security. Let $\{y_n\}_{n\in \N}$ be an infinite family of input sequences where $|y_n|=n$, for each $n\in \N$.
    
	There are constants $c_0, c_1 > 0$, such that for any $n \geq c_0$, $\M_{n}$ must perform in expectation at least $c_1 n \log n$ server probes on the input sequence $y_n$.
	\label{thm:main_strong_computational}
\end{theorem}

\begin{proof}
    The proof is identical to the proof of Theorem~\ref{thm:main} but we use Lemma~\ref{lemma:graph_properties_strong_comp_sec} instead of Lemma~\ref{lemma:graph_properties_statistical_sec}.
    Note that the different order of quantifiers is caused by different order of quantifiers in Lemma~\ref{lemma:graph_properties_statistical_sec} and in Lemma~\ref{lemma:graph_properties_strong_comp_sec}.
\end{proof}

In the rest of this section, we prove an $\omega(1)$ lower bound for ORAMs satisfying weak computational security from Definition~\ref{def:comp_sec_weak}.
Note that in the case of weak computational security it is unclear which $k$ should the adversary use to distinguish $y$ and $y'$.
Thus, we cannot directly conclude that $y$ has $\frac{n}{5k}$-dense $k$-partition for every $n$ and $k \leq \left\lfloor \frac{n}{10(m(n)+2 \log n + 11)} \right\rfloor$.
On the other hand, for every $k$ there could be only finitely many values $n$ such that there is an input sequence of length $n$ which has no $\frac{n}{5k}$-dense $k$-partition.
This fact allows us to prove the $\omega(1)$ lower bound for weak computational security.

\begin{theorem}
    Let $m, M, w \colon \mathbb{N} \rightarrow \mathbb{N}$ be non-decreasing functions such that for all $n$ large enough: $m(n) \leq \sqrt{n}$ and $n \leq M(n) \leq 2^{w(n)}$.
    Let $\{ \M_{n} \}_{n\in \N}$ be a sequence of online ORAMs with address range $M(n)$, cell size $w(n)$ bits and $m(n)$ cells of internal memory which satisfy weak computational security. Let $\{y_n\}_{n\in \N}$ be a sequence of input sequences where $|y_n|=n$, for each $n\in \N$.
    
	For any constant $c_1 > 0$ there is a constant $c_0 > 0$, such that for any $n \geq c_0$, $\M_{n}$ must perform in expectation at least $c_1 n$ server probes on the input sequence $y_n$.
	
	\label{thm:mainomega}
\end{theorem}

	In particular there is no computationally secure online ORAM with constant bandwidth overhead $\mathcal{O}(1)$.

\begin{proof}
    For each $n\in\N$, define $k(n)$ to be the smallest $k$ such that $$\Pr[G(A(\M_{n},y_n)) \text{ has  $(n/5k)$-dense $k$-partition}] < 1/2.$$
    Using Corollary~\ref{cor:main_graph_cor} we get for each $n$ large enough that the expected number of edges in $G(A(\M_{n},y_n))$ is at least $c \cdot n \log k(n)$, for some absolute constant $c>0$.
    It suffices to show that $k(n) \rightarrow \infty$ as $n \rightarrow \infty$.
	There cannot exist a constant $k$ such that $Y_n$ has $(n/5k)$-dense $k$-partition with probability less than $\frac{1}{2}$ for infinitely many $n$.
	Otherwise $\left\{ y_n \right\}_n$ would be computationally distinguishable from $\left\{ Y_{n, k} \right\}_n$ (by the greedy algorithm which has $k$ hard-wired). So,  $k(n) \rightarrow \infty$ as $n \rightarrow \infty$.
\end{proof}

\section{Alternative Definitions for Oblivious RAM}\label{sec:AlternativeDefinitions}
In this section, we recall some alternative definitions for ORAM which appeared in the literature and explain the relation of our lower bound to those models.

\paragraph{The definition of Larsen and Nielsen.}
Larsen and Nielsen (see Definition 4 in \cite{LarsenN18}) required that for any two input sequences of equal length, the corresponding distributions of access sequences cannot be distinguished with probability greater than $ \sfrac{1}{4} $ by any algorithm running in polynomial time in the sum of the following terms: the length of the input sequence, logarithm of the number of memory cells (i.e., $ \log n $), and the size of a memory cell (i.e., $ \log n $ for the most natural parameters).
We show that their definition implies statistical closeness as considered in our work (see~the statistical security property in Definition~\ref{def:oram}).
Therefore, any lower bound on the bandwidth overhead of ORAM satisfying our definition implies a matching lower bound w.r.t. the definition of Larsen and Nielsen~\cite{LarsenN18}.

To this end, let us show that if two distributions of access sequences are not statistically close, then they are distinguishable in the sense of Larsen and Nielsen.
Assume there exist two input sequences $ y$ and  $y' $ of equal lengths, for which the access sequences $ A(\M, y) $ and $ A(\M, y') $ have statistical distance greater than $ \sfrac{1}{4} $.
We define a distinguisher algorithm $ D $ that on access sequence $ x $ outputs $ 1 $ whenever $ \Pr[A(\M, y)=x]> \Pr[A(\M, y')=x]$, outputs $ 0 $ whenever $ \Pr[A(\M, y)=x]< \Pr[A(\M, y')=x]$, and outputs a uniformly random bit whenever $ \Pr[A(\M, y)=x]= \Pr[A(\M, y')=x]$.
It follows from definition of $D$, basic properties of statistical distance (see~\cref{prop:StatDistanceAsArea}), and our assumption about the statistical distance of $ A(\M, y) $ and $ A(\M, y') $ that
\[
|\Pr[D(A(\M, y))=1]-\Pr[D(A(\M, y'))=1]| = \SD{A(\M, y)}{A(\M, y')}>\frac{1}{4}\ .
\]
Note that $ D $ can be specific for the pair of the two input sequences $y$ and $y'$ and it can have all the significant information about the distributions $ A(\M, y) $ and $ A(\M, y') $ hardwired.
For example, it is sufficient to store a string describing for each access sequence $x$ whether it is more, less, or equally likely under $A(\M, y)$ or $A(\M, y')$.
Even though such string is of exponential size w.r.t. the length of the access pattern, $D$ needs to simply access the position corresponding to the observed access pattern to output its decision as described above.
Thus, $ D $ can run in linear time in the length of the access sequence (which is polynomial in the length of the input sequence) and distinguishes the two access sequences with probability greater than $ \sfrac{1}{4} $.

\paragraph{The definition of Goldreich and Ostrovsky.} 
Unlike the original definition of ORAM from Goldreich~\cite{Goldreich87} and Ostrovsky~\cite{Ostrovsky90}, the definition of ORAM presented in Goldreich and Ostrovsky~\cite{GoldreichO96} postulates an alternative security requirement.
However, the alternative definition suffers from an issue which is not present in the original definition and which, to the best of our knowledge, was not pointed out in the literature.
In particular, the definition in~\cite{GoldreichO96} can be satisfied by a dummy ORAM construction with only a constant overhead and without achieving any indistinguishability of the access sequences.
Given that Goldreich and Ostrovsky~\cite{GoldreichO96} might serve as a primary reference for our community, we explain the issue in the following paragraph to help preventing the use of the problematic definition in future works.

Recall the definition of ORAM with perfect security from Goldreich and Ostrovsky (Definition 2.3.1.3 in \cite{GoldreichO96}):

\medskip
\noindent\textbf{Goldreich-Ostrovsky security:}\emph{ For any two input sequences $y$ and $y'$, if the length distributions $|A(\M, y)|$ and $|A(\M, y')|$ are identical, then $A(\M, y)$ and $A(\M, y')$ are identical.}

\medskip
As we show, this requirement can be satisfied by creating an ORAM that makes sure that on any two distinct sequences $y,y'$, the length distributions $|A(\M, y)|$ and $|A(\M, y')|$ differ. Note that no indistinguishability is required in that case and the ORAM can then reveal the access pattern of the input sequence.

To this end, we describe an ORAM with a constant overhead so that $|A(\M, y)| \in \{2|y|, 2|y|+1\}$ and the distribution $|A(\M, y)|$ encodes the sequence $y$.
The ORAM proceeds by performing every operation $y_i$ directly on the server followed by a read operation from address 1.
After the last instruction in $y$, the ORAM selects a random sequence of operations $r$ of length $|y|$ and if $r$ is lexicographically smaller than $y$ then the ORAM performs an extra read from address 1 before terminating.
Note that this ORAM can be efficiently implemented using constant amount of internal memory by comparing the input sequence to the randomly selected one online. 
Also, the machine does not need to know the length of the sequence in advance.
Finally, the length distribution $ |A(\M, y)| $ is clearly different for each input sequence $ y $.
Given that the above definition of ORAM of Goldreich and Ostrovsky allows the dummy construction with a constant overhead, we do not hope to extend our lower bound towards this definition.

One could object that the above dummy ORAM exploits the fact that indistinguishability of access sequences must hold only if the length distributions are identical.
However, it is possible to construct a similar dummy ORAM with low overhead satisfying even the following relaxation of the definition requiring indistinguishability of access sequences corresponding to any pair of $y$ and $y'$ for which $|A(M,y)|$ and $|A(M,y')|$ are statistically close (i.e., the indistinguishability is required for a potentially larger set of access patterns):

\bigskip
\noindent \textbf{Relaxation of Goldreich-Ostrovsky security:} \emph{For any two input sequences $y$ and $y'$, if the length distributions $|A(\M, y)|$ and $|A(\M, y')|$ are statistically close, then $A(\M, y)$ and $A(\M, y')$ are statistically close.}

\medskip
We show there is a dummy ORAM $\M$ with a constant overhead such that for any two input sequences $y$ and $y'$ which differ in their accessed memory locations, the statistical distance $\SD{|A(\M, y)|}{|A(\M, y')|}$ is at least $\frac{1}{n M}$ (where $n=|y|=|y'|$ and $M$ is the size of address range).

The ORAM $\M$ works as follows.
At the beginning, the ORAM picks $i\in [n]$ and $r\in [M]$ uniformly at random.
Then for $j=1,\dots n$, it executes each of the input operations $(o_j, a_j, d_j)$ directly on the server.
For each $j<i$, it performs two additional reads from address 1 after executing the $j$-th input operation.
For $j=i$, after the $i$-th input operation it performs two additional reads from address 1 if $r \le a_i$, and it performs one additional read from address 1 if $r>a_i$. 
For $j>i$, it performs each of the input operations without any additional read.

It is straightforward to verify that the distribution of $|A(\M, y)|$ satisfies: for each $i\in[n]$, $\Pr[|A(\M, y)|=n+2i]=\frac{a_i}{nM}$.
Hence, for any pair $y$ and $y'$ of two input sequences of length $n$, if the sequences of addresses accessed by them differ then 
the statistical distance between the distributions of $|A(\M, y)|$ and $|A(\M, y')|$ is at least $1/nM$. If $M$ is polynomial in $n$ this means that their distance is at least $\frac{1}{\mathrm{poly}(n)}$. 
%
Thus, $\M$ satisfies even the stronger variant of the definition from \cite{GoldreichO96} even though its access sequence leaks the addresses from the input sequence.

It was previously shown by Haider, Khan and van~Dijk~\cite{HaiderD17} that there exists an ORAM construction which reveals all memory accesses from the input sequence while satisfying the definition of Goldreich and Ostrovsky from~\cite{GoldreichO96}.
However, their construction has an \emph{exponential} bandwidth overhead which makes it insufficient to demonstrate any issue with the definition of Goldreich and Ostrovsky.
Clearly, any definition of ORAM can disregard constructions with super-linear overhead as a perfectly secure ORAM (with linear overhead) can be constructed by simply passing over the whole server memory for each input operation.
Unlike the construction of \cite{HaiderD17}, our constructions of the dummy ORAMs with constant bandwidth overhead exemplify that the definition of Goldreich and Ostrovsky from~\cite{GoldreichO96} is problematic in the interesting regime of parameters.

\paragraph{Simulation-based definitions.}
The recent work of Asharov~et~al.~\cite{AsharovKLNS18} employs a simulation-based definition parameterized by a functionality which implements an oblivious data structure.
Our lower bounds directly extend to their stronger definition when the functionality implements Array Maintenance.
Moreover, our techniques can be adapted to give lower bounds for functionalities implementing stacks, queues and others considered in \cite{JacobLN19}.

\paragraph{Weak vs. strong computational security.}
In this work, we distinguish between weak and strong computational security (see Definition~\ref{def:oram}).
Our techniques do not allow to prove matching bounds for ORAMs satisfying the two notions and we show $\Omega(\log n)$ lower bound only w.r.t. strong computational security. 
Though, as we noted in Section~\ref{sec:OurResults}, even the $\omega(1)$ lower bound for online ORAMs satisfying weak computational security is an interesting result in the light of the work of Boyle and Naor~\cite{BoyleN16}.
It follows from~\cite{BoyleN16} that any super-constant lower bound for \emph{offline} ORAM would imply super-linear lower bounds on size of sorting circuits -- which would constitute a major breakthrough in computational complexity.
The main result from Boyle and Naor \cite{BoyleN16} can be rephrased using our notation as follows.
\begin{theorem}[Theorem 3.1 \cite{BoyleN16}]
Suppose there exists a Boolean circuit ensemble $C = \{C( n,w)\}_{n,w}$
of size $s(n,w)$, such that each $C(n,w)$ takes as input $n$ words each of size
$w$ bits, and outputs the words in sorted order.  Then  for word  size $w \in \Omega(\log n) \cap n^{o(1)}$ 
and constant internal memory $m \in \mathcal{O}(1)$, there  exists a secure offline ORAM
(as per Definition 2.8~\cite{BoyleN16}) with total bandwidth and computation $\mathcal{O}(n \log w + s(2 n/w,w))$.
\end{theorem}
Moreover, the additive factor of $\mathcal{O}(n \log w)$ follows from the transpose part of the algorithm of \cite{BoyleN16} (see Figures~1 and 2 in \cite{BoyleN16}).
As Boyle and Naor showed in their appendix (Remark B.3~\cite{BoyleN16}) this additive factor in total bandwidth may be reduced to $\mathcal{O}(n)$ if the size of internal memory is $m \geq w$.
Thus, sorting circuit of size $\mathcal{O}(nw)$ implies offline ORAM with total bandwidth $\mathcal{O}(n+2\frac{n}{w}w) = \mathcal{O}(n)$.
Or the other way around, lower bound $\omega(n)$ for total bandwidth of offline ORAM implies $\omega(nw)$ lower bound for circuits sorting $n$ words of size $w$ bits, each.

We leave it as an intriguing open problem whether it is possible to prove an $ \Omega(\log n) $ lower bound for online ORAMs satisfying weak computational security.

\section*{Acknowledgements}
We wish to thank Oded Goldreich for clarifications regarding the ORAM definitions in \cite{Goldreich87,Ostrovsky90,GoldreichO96} and Jesper Buus Nielsen for clarifying the details of the lower bound for computationally secure ORAMs from \cite{LarsenN18}.
We are also thankful to the anonymous TCC 2019 reviewers for insightful comments that helped us improve the presentation of our results.

\newcommand{\etalchar}[1]{$^{#1}$}

\end{document}